\documentclass[11pt]{article}
\usepackage{amsmath, amssymb,amsthm}
\usepackage{hyperref}
\usepackage{geometry}
\usepackage{cite}

\newcommand{\ket}[1]{|#1\rangle}

\newcommand{\supp}{\operatorname{supp}}

\newtheorem{theorem}{Theorem}
\newtheorem{corollary}[theorem]{Corollary}

\begin{document}

\title{From Joint to Single-System $\psi$-Onticity \\ Without Preparation Independence} 

\author{Shan Gao
\\Research Center for Philosophy of Science and Technology, 
\\ Shanxi University, Taiyuan 030006, P. R. China
\\ E-mail: \href{mailto:gaoshan2017@sxu.edu.cn}{gaoshan2017@sxu.edu.cn}.}

\maketitle

\begin{abstract}
The Pusey–Barrett–Rudolph (PBR) theorem establishes $\psi$-onticity for individual quantum systems, but its standard formulation relies on the Preparation Independence Postulate (PIP). This has led to a prevalent view that rejecting PIP leaves open the possibility of $\psi$-epistemic models for individual systems. In this work, we show that this understanding is incomplete: once the PBR theorem establishes $\psi$-onticity for composite systems prepared in product states, the $\psi$-onticity of the individual subsystems follows directly from the tensor-product structure of quantum mechanics, without invoking PIP or any further auxiliary assumptions. This result removes a key auxiliary assumption from the PBR theorem, closes a persistent loophole for preserving $\psi$-epistemic models, and strengthens the conceptual foundations of $\psi$-ontology.
\end{abstract}

\section{Introduction}

The question of whether the quantum wave function corresponds to a physical property of an individual system, or merely represents an experimenter's knowledge, lies at the heart of interpreting quantum mechanics \cite{Author2017}. Within the ontological models framework \cite{Spekkens2010}, this distinction is formalized by the concepts of $\psi$-onticity---where distinct quantum states correspond to disjoint distributions over ontic states---and $\psi$-epistemicity, where such distributions may overlap. The Pusey--Barrett--Rudolph (PBR) theorem stands as a landmark result in this debate, widely regarded as establishing $\psi$-onticity for single quantum systems \cite{PBR2012}. However, its standard formulation relies on an auxiliary assumption known as the Preparation Independence Postulate (PIP), which posits that independently prepared systems possess uncorrelated ontic states.

This reliance on PIP has shaped the prevailing understanding of the theorem's scope. The PBR argument proceeds in two logical steps: first, it establishes $\psi$-onticity for composite systems prepared in distinct product states within the ontological models framework; second, PIP is invoked to extend this conclusion to the individual subsystems. Consequently, it has been widely held that rejecting PIP could permit $\psi$-epistemic models for single systems to survive, even while respecting the PBR result for composite systems. This view has motivated several proposed models that relax PIP in an attempt to retain $\psi$-epistemicity at the subsystem level \cite{Lewis2012, Aaronson2013, Leifer2014}. 

In this paper, we show that this understanding is incomplete. 
We demonstrate that once the $\psi$-onticity for product states is established, the $\psi$-onticity of the individual subsystem wave functions follows directly from the tensor-product structure of quantum mechanics, with no appeal to PIP or any other auxiliary assumptions. The key conceptual observation is that $\psi$-onticity itself has a powerful structural consequence: if the ontic supports for distinct pure states never overlap, then one can always reorganize the ontic space so that each ontic state carries a sharply defined label identifying the quantum state it is compatible with. This yields a natural decomposition of the ontic state into a ``$\psi$-component’’ and residual degrees of freedom or hidden variables, and guarantees that every $\psi$-ontic model is equivalent to one in which the preparation measure takes the form $\mu_\psi(\lambda_\psi,\eta)=\delta(\lambda_\psi-\psi)\,\nu_\psi(\eta)$. 
This is not an additional assumption: it is a representation theorem that follows directly from $\psi$-onticity rather than supplementing it.

By applying this representation to the $\psi$-ontic joint state $\psi_1 \otimes \psi_2$, we show that the $\psi$-components associated with the subsystems must themselves be uniquely fixed, rendering $\psi_1$ and $\psi_2$ each $\psi$-ontic individually. Hidden-variable correlations, however strong, cannot alter this conclusion. Our result closes the loophole often thought to arise if PIP is rejected: joint $\psi$-ontology already entails single-system $\psi$-ontology. The remainder of the paper develops this argument in detail, clarifies the status of models claiming to evade the PBR conclusion, and examines the implications for realist interpretations of the quantum state.

\section{The PBR Theorem: Joint $\psi$-Onticity and PIP}
\label{sec:pbr}

In ontological models of quantum mechanics, a physical system that can be assigned to a wave function or a pure state is described by an ontic state $\lambda \in \Lambda$. 
A pure quantum state $\ket{\psi}$ is said to be $\psi$-ontic if the epistemic distributions $\mu_\psi(\lambda)$ for distinct quantum states have disjoint supports.
That is, for $\ket{\psi}\neq\ket{\phi}$,
\begin{equation}
\mathrm{supp}(\mu_\psi) \cap \mathrm{supp}(\mu_\phi) = \varnothing.
\end{equation}
Disjointness means that the ontic state uniquely determines the quantum state: $\lambda$ never arises from two different preparations.
Hence the wave function is an objective physical property of the individual system.
If the supports overlap, then a single ontic state may arise from different quantum states; in this case the wave function is $\psi$-epistemic \cite{Leifer2014}.

\subsection{Structural Decomposition of $\psi$-Ontic Models}

A central observation underlying our analysis is that $\psi$-onticity has a strong structural consequence. Whenever the supports of $\{\mu_\psi\}$ are pairwise disjoint, the ontic space $\Lambda$ can be reorganized (up to sets of measure zero) into
\begin{equation}
\lambda = (\lambda_\psi, \eta),
\end{equation}
where
\begin{itemize}
\item $\lambda_\psi \in \mathbb{CP}^{d-1}$ is a sharp label identifying the prepared quantum state,
\item $\eta$ denotes any residual ontic degrees of freedom.
\end{itemize}
In this coordinate system, every $\psi$-ontic model takes the delta-function form
\begin{equation}
\mu_\psi(\lambda_\psi,\eta)
= \delta(\lambda_\psi - \psi)\, \nu_\psi(\eta).
\label{eq:psi-ontic-form}
\end{equation}
Here $\nu_\psi(\eta)$ is a normalized probability measure over $\eta$, and it may in general depend on $\psi$. 
This decomposition is not an additional assumption, but a representation theorem: pairwise disjoint supports necessarily imply that the quantum state is explicitly and sharply encoded inside the ontic state. A complete proof is provided in Appendix~\ref{app:representation}.

This formalizes what PBR explicitly state in their original paper:
\begin{quote}
``If the same [i.e. the distributions cannot overlap] can be shown for any pair of quantum states $\ket{\psi_0}$ and $\ket{\psi_1}$, then the quantum state can be inferred uniquely from $\lambda$. In this case, the quantum state is a physical property of the system.'' \cite{PBR2012}
\end{quote}
Their intended meaning is precisely the delta-function representation in Eq.~\eqref{eq:psi-ontic-form}. 
In their framework, if the ontic distributions corresponding to distinct quantum states are disjoint, then the label of the quantum state is uniquely determined by the ontic state, and the quantum state is a physical property of the system.  
By introducing the measurable map $f$ and defining $\lambda_\psi := f(\lambda)$, and separating the residual degrees of freedom into $\eta$ (see Appendix~\ref{app:representation}), we make this expression mathematically rigorous: the wave function is explicitly encoded in the ontic state, exactly as PBR describe when they state that ``every detail of the quantum state is `written into' the real physical state of affairs." \cite{PBR2012} In this sense, our delta-function representation is a formal expression of PBR’s original concept of $\psi$-onticity.

\subsection{Joint $\psi$-Onticity in the PBR Scenario}

The PBR theorem considers two independently preparable systems with Hilbert spaces $\mathcal{H}_1$ and $\mathcal{H}_2$, prepared in product states $\ket{\psi_1}\otimes\ket{\psi_2}$.
PBR prove that within the ontological models framework the joint quantum state is $\psi$-ontic: distinct product states
\begin{equation}
\ket{\psi_1\otimes\psi_2}\neq\ket{\phi_1\otimes\phi_2}
\end{equation}
yield disjoint epistemic distributions in the composite ontic space $\Lambda$.
By the representation theorem above, this implies
\begin{equation}
\mu_{\psi_1\otimes\psi_2}(\lambda)
= \delta(\lambda_\psi - \psi_1\otimes\psi_2)\,
\nu_{\psi_1\otimes\psi_2}(\eta).
\label{eq:joint-delta}
\end{equation}
This expression is not an assumption about the structure of the composite system.
It is simply the formal statement of what ``the joint state is $\psi$-ontic'' means.
Crucially, no assumption has been made about $\psi$-onticity of the subsystems individually.
Only the composite system’s $\psi$-onticity is used.

\subsection{The Role of PIP}

In the original PBR argument, the inference from joint $\psi$-onticity to single-system $\psi$-onticity relies on the Preparation Independence Postulate (PIP).\footnote{Note that in the PBR argument, the joint $\psi$-onticity is derived entirely at the level of the composite system and does not rely on preparation independence for the subsystems; this assumption is introduced only in the subsequent step that infers single-system $\psi$-onticity from joint $\psi$-onticity \cite{PBR2012,Leifer2014}.} PIP asserts that independently prepared systems possess uncorrelated ontic states, so that the epistemic distribution for a product preparation factorizes:
\begin{equation}
\mu_{\psi_1 \otimes \psi_2}(\lambda_1,\lambda_2)
= \mu_{\psi_1}(\lambda_1)\,\mu_{\psi_2}(\lambda_2).
\end{equation}
If a model permits single-system $\psi$-epistemicity---for example, if $\ket{0}$ and $\ket{+}$ overlap in ontic space---then PIP transmits these overlaps to the composite level. Under PIP, the four product-state preparations used in the PBR experiment, 
$\ket{0\otimes 0}$, $\ket{0\otimes +}$, $\ket{+\otimes 0}$, and $\ket{+\otimes +}$, necessarily yield joint distributions with nonzero mutual overlap. Such overlaps contradict the joint $\psi$-onticity required by the PBR measurement statistics. Consequently, within the original PBR framework, PIP is precisely the mechanism that rules out the possibility that subsystem overlaps could be ``washed away'' by correlations, thereby forcing single-system $\psi$-ontology (for further analysis of PIP see \cite{Schlosshauer2012, Schlosshauer2014, Leifer2014}).  

The widely received view holds that the PBR theorem’s joint $\psi$-onticity does not necessarily imply $\psi$-onticity for single systems if PIP is relaxed, allowing correlated ontic states across subsystems to permit $\psi$-epistemicity \cite{Lewis2012, Aaronson2013, Leifer2014}. Models like that of Lewis et al. suggest that $\psi$-epistemicity for single systems is possible by introducing such correlations, though they did not show that their model fully reproduces quantum mechanics’ entanglement measurement predictions \cite{Lewis2012}. 

\subsection{Our Contribution: Single-System $\psi$-Onticity from Joint $\psi$-Onticity without PIP}

The next section revisits and clarifies this commonly held view. 
We show that, given the PBR theorem’s first result that product states are jointly $\psi$-ontic and the tensor-product structure of quantum mechanics, one can derive single-system $\psi$-onticity without appealing to PIP or any further assumptions.
Correlations in the ontic state of the joint system do not rescue $\psi$-epistemic models once the structural form~\eqref{eq:psi-ontic-form} is correctly acknowledged as a consequence of joint $\psi$-onticity.
Thus, the standard PBR reasoning can be completed with strictly fewer assumptions than previously thought.

\section{A Direct Proof of Single-System $\psi$-Ontology}
\label{sec:proof}

We now prove that the individual subsystem wave functions
$\ket{\psi_1}\in\mathcal{H}_1$ and $\ket{\psi_2}\in\mathcal{H}_2$
are themselves $\psi$-ontic, using only two ingredients:
\begin{enumerate}
  \item PBR's result that the joint product state $\ket{\psi_1 \otimes \psi_2}$ is $\psi$-ontic~\cite{PBR2012}, and
  \item the tensor-product structure of the Hilbert space, which identifies $\ket{\psi_1}$ and $\ket{\psi_2}$ as the factors of the joint wave function.
\end{enumerate}
No preparation-independence assumption and no constraints on hidden-variable correlations are required.

\subsection{Consequence of joint $\psi$-onticity}

The PBR theorem implies that for every product state $\ket{\psi_1 \otimes \psi_2}$, the corresponding epistemic distribution on the composite ontic space $\Lambda$ has the form
\begin{equation}
\mu_{\psi_1\otimes\psi_2}(\lambda)
= \delta(\lambda_\psi - \psi_1\otimes\psi_2)\,
\nu_{\psi_1\otimes\psi_2}(\eta),
\label{eq:joint}
\end{equation}
where we have parameterized the ontic state as $\lambda = (\lambda_\psi,\eta)$ and $\lambda_\psi$ is the component that uniquely identifies the prepared joint quantum state (see Appendix~\ref{app:representation}).
Equation~\eqref{eq:joint} provides a clear mathematical formulation of joint $\psi$-onticity: no two distinct product states ever assign positive probability to the same value of $\lambda_\psi$.

\subsection{Tensor decomposition of the $\psi$-label}

Because the quantum state of the composite system factorizes as
\begin{equation}
\psi_1 \otimes \psi_2 \in \mathcal{H}_1 \otimes \mathcal{H}_2,
\end{equation}
the label $\lambda_\psi$ appearing in~\eqref{eq:joint} can itself be coordinatized in the natural tensor-product way:
\begin{equation}
\lambda_\psi = (\lambda_{\psi_1}, \lambda_{\psi_2})
\qquad\text{with}\qquad
\lambda_{\psi_1} \in \mathbb{CP}^{d_1-1},
\quad \lambda_{\psi_2} \in \mathbb{CP}^{d_2-1}.
\end{equation}
Since $\lambda_\psi = \psi_1 \otimes \psi_2$ almost everywhere under $\mu_{\psi_1\otimes\psi_2}$, the unique specification of the joint state forces the unique specification of its tensor factors:
\begin{equation}
\lambda_{\psi_1} = \psi_1,
\qquad
\lambda_{\psi_2} = \psi_2
\end{equation}
(up to sets of measure zero). Consequently, the delta distribution factorizes as
\begin{equation}
\delta(\lambda_\psi - \psi_1\otimes\psi_2)
= \delta(\lambda_{\psi_1} - \psi_1)\,
\delta(\lambda_{\psi_2} - \psi_2).
\label{eq:deltaprod}
\end{equation}
This factorization is \emph{not} an additional assumption about the ontic space; it is the canonical way of expressing the already-established joint $\psi$-onticity in the natural tensor coordinates dictated by the Hilbert-space factorization.

To make the above structural consequence fully rigorous, we provide a short argument by contradiction showing that $\mu_{\psi_1}$ and $\mu_{\psi_2}$ must each be delta distributions concentrated on the corresponding $\psi$-labels. 
Assume, for contradiction, that the subsystem distribution $\mu_{\psi_1}(\lambda_{\psi_1})$ is \emph{not} a delta function. Then its support contains at least two distinct values
\begin{equation}
\lambda_{\psi_1}^a \neq \lambda_{\psi_1}^b
\end{equation}
that both occur with positive probability under preparation of $\ket{\psi_1}$.
When the composite state $\ket{\psi_1 \otimes \psi_2}$ is prepared, the corresponding ontic state of the full system could then be either
\begin{equation}
\lambda^{(1)} = \bigl( (\lambda_{\psi_1}^a, \lambda_{\psi_2}^a), \eta^{(1)} \bigr),
\qquad
\lambda^{(2)} = \bigl( (\lambda_{\psi_1}^b, \lambda_{\psi_2}^b), \eta^{(2)} \bigr),
\end{equation}
where $\lambda_{\psi_2}^{a,b}$ and $\eta^{(1),(2)}$ are the $\psi$-related ontic states of subsystem 2 and the residual degrees of freedom of the composite system. Both $\lambda^{(1)}$ and $\lambda^{(2)}$ occur with positive probability under $\mu_{\psi_1 \otimes \psi_2}$.
However, joint $\psi$-onticity requires that for the global state
\begin{equation}
\lambda_\psi^{(1)} = \lambda_\psi^{(2)} = \psi_1 \otimes \psi_2.
\end{equation}
By definition of the $\psi$-related component (as constructed in Theorem 1),
\begin{equation}
\lambda_\psi^{(1)} = (\lambda_{\psi_1}^a, \lambda_{\psi_2}^a), 
\qquad \lambda_\psi^{(2)} = (\lambda_{\psi_1}^b, \lambda_{\psi_2}^b).
\end{equation}
Thus
\begin{equation}
(\lambda_{\psi_1}^a, \lambda_{\psi_2}^a) = (\lambda_{\psi_1}^b, \lambda_{\psi_2}^b),
\end{equation}
and hence
\begin{equation}
\lambda_{\psi_1}^a = \lambda_{\psi_1}^b,
\qquad
\lambda_{\psi_2}^a = \lambda_{\psi_2}^b,
\end{equation}
contradicting the assumption $\lambda_{\psi_1}^a \neq \lambda_{\psi_1}^b$.
Therefore,
\begin{equation}
\mu_{\psi_1}(\lambda_{\psi_1}) = \delta(\lambda_{\psi_1} - \psi_1), 
\qquad \mu_{\psi_2}(\lambda_{\psi_2}) = \delta(\lambda_{\psi_2} - \psi_2).
\end{equation}
This establishes that~\eqref{eq:deltaprod} is the only consistent form of the composite distribution.

Substituting~\eqref{eq:deltaprod} into~\eqref{eq:joint} yields
\begin{equation}
\mu_{\psi_1 \otimes \psi_2}(\lambda)
= \delta(\lambda_{\psi_1} - \psi_1) \;
  \delta(\lambda_{\psi_2} - \psi_2) \;
  \nu_{\psi_1 \otimes \psi_2}(\eta).
\label{eq:jointfactor}
\end{equation}
Hidden variables and their possible correlations are entirely captured by $\nu_{\psi_1 \otimes \psi_2}(\eta)$; they play no role in the logical necessity of the $\psi$-delta structure.
Joint $\psi$-onticity prohibits multiple $\psi$-labels for the same quantum state, even though it permits epistemic uncertainty in the remaining degrees of freedom.

\subsection{Extraction of the subsystem distributions}

Consider subsystem~1. Its epistemic distribution is the marginal
\begin{equation}
\mu_{\psi_1}(\lambda_{\psi_1})
= \int d\lambda_{\psi_2}\, d\eta\,
\mu_{\psi_1\otimes\psi_2}(\lambda).
\end{equation}
Inserting~\eqref{eq:jointfactor} gives directly
\begin{equation}
\mu_{\psi_1}(\lambda_{\psi_1})
= \delta(\lambda_{\psi_1} - \psi_1).
\label{eq:sub1}
\end{equation}
The identical argument for subsystem~2 yields
\begin{equation}
\mu_{\psi_2}(\lambda_{\psi_2})
= \delta(\lambda_{\psi_2} - \psi_2).
\label{eq:sub2}
\end{equation}
Thus, the $\psi$-component of each subsystem’s ontic state is uniquely and sharply fixed by the corresponding subsystem wave function. Any additional hidden variables (possibly correlated between subsystems) are confined to the residual variable $\eta$ and do not affect the $\psi$-label.

\subsection{Single-system $\psi$-ontology}

To state the conclusion explicitly: suppose $\ket{\psi_1} \neq \ket{\phi_1}$.
Then~\eqref{eq:sub1} implies
\begin{equation}
\mathrm{supp}(\mu_{\psi_1}) = {\lambda_{\psi_1} = \psi_1},
\qquad
\mathrm{supp}(\mu_{\phi_1}) = {\lambda_{\phi_1} = \phi_1}.
\end{equation}
These supports are disjoint because $\psi_1 \neq \phi_1$.
Hence subsystem~1 is $\psi$-ontic.
The same reasoning applies independently to subsystem~2.

This completes the proof: joint $\psi$-onticity of product states, together with the tensor-product structure of the quantum state space, \emph{already implies} single-system $\psi$-onticity. This conclusion requires neither PIP nor any assumptions about the statistical independence of hidden variables. Correlations—however strong—may exist arbitrarily within the non-$\psi$ degrees of freedom encoded by $\eta$, but they cannot alter the unique specification of the $\psi$-related component of the ontic state.

In summary, the PBR-style disjointness of composite-state epistemic distributions, when expressed in the natural tensor coordinates of the Hilbert space, forces the $\psi$-label of each subsystem to be uniquely fixed by its quantum state. Consequently, $\ket{\psi_1}$ and $\ket{\psi_2}$ are genuine physical properties of their respective subsystems, establishing single-system $\psi$-ontology in a direct and rigorous way.

\section{Limitations of $\psi$-Epistemic Models}
\label{sec:lewis}

Attempts to construct $\psi$-epistemic models, where distinct quantum states share ontic states, have been proposed to challenge the $\psi$-onticity of single-system wave functions. Two notable models, by Lewis et al. \cite{Lewis2012} and Aaronson et al. \cite{Aaronson2013}, aim to permit $\psi$-epistemicity for single systems by relaxing PIP. However, their limitations, particularly in addressing composite systems and entanglement measurements, mean they do not challenge our proof of subsystem $\psi$-ontology, which relies on the PBR theorem’s $\psi$-onticity for composite systems and the tensor product structure of quantum mechanics.

\subsection{Lewis et al. Model}
The Lewis et al. model defines an ontic state space $\Lambda = \mathbb{CP}^{d-1} \times [0,1]$, where $\ket{\lambda} \in \mathbb{CP}^{d-1}$ represents the $\psi$-related part (equivalent to the quantum state space, e.g., $S^2$ for qubits) and $x \in [0,1]$ is a hidden variable. For qubits, a preferred state $\ket{0}$ (north pole on the Bloch sphere) defines a hemisphere $\mathcal{R}_0$ ($\theta_\lambda < \pi/2$) and a subset $\mathcal{E}_0 = \{ (\hat{\lambda}, x) : \hat{\lambda} \in \mathcal{R}_0, 0 \leq x < (1 - \sin \theta_\lambda)/2 \}$, where $\theta_\lambda$ is the angle from $\ket{0}$. The epistemic state for a quantum state $\ket{\psi} \in \mathcal{R}_0$ is given by:
\begin{equation}
\mu_\psi(\hat{\lambda}, x) = \delta(\hat{\lambda} - \hat{\psi}) \Theta\left( x - \frac{1 - \sin \theta_\psi}{2} \right) + \frac{1 - \sin \theta_\psi}{2} \mu_{\mathcal{E}_0}(\hat{\lambda}, x),
\end{equation}
where $\mu_{\mathcal{E}_0}$ is a distribution over $\mathcal{E}_0$. This allows distinct states $\ket{\psi}, \ket{\phi} \in \mathcal{R}_0$ to share ontic states with same $\hat{\lambda} \in \mathcal{R}_0$, achieving $\psi$-epistemicity. The response function, $\xi_{\phi_k}(\hat{\lambda}, x) = \Theta \left[ (|\langle \lambda | \phi_0 \rangle|^2 - x) (-1)^k \right]$, ensures the Born rule for single-system projective measurements, where measurements are ordered relative to $\ket{0}$ (e.g., $|\langle \phi_0 | 0 \rangle|^2 \geq |\langle \phi_1 | 0 \rangle|^2$).

The model is explicitly constructed for single systems, reproducing quantum mechanics’ Born rule for projective measurements on these systems. However, the PBR theorem, which our proof relies upon, leverages quantum mechanics’ predictions for entanglement measurements in composite systems to establish $\psi$-onticity \cite{PBR2012}. The Lewis et al. model does not provide a detailed framework for composite systems or specify response functions for joint measurements, particularly those involving entanglement, such as Bell-basis measurements critical to the PBR theorem.

The PBR theorem considers composite states like $\ket{0} \otimes \ket{0}$, $\ket{0} \otimes \ket{+}$, $\ket{+} \otimes \ket{0}$, and $\ket{+} \otimes \ket{+}$, measured in an entangled basis (e.g., $\frac{1}{\sqrt{2}} (\ket{0} \ket{1} - \ket{1} \ket{0})$). If epistemic distributions overlap, a $\psi$-epistemic model predicts non-zero probabilities for outcomes quantum mechanics assigns zero probability, leading to a contradiction. In the Lewis et al. model, distinct states like $\ket{0}$ and $\ket{+}$ share ontic states in $\mathcal{E}_0$ with the same $\hat{\lambda} \in \mathcal{R}_0$. If extended to a composite system, the epistemic states for $\ket{0} \otimes \ket{0}$ and $\ket{0} \otimes \ket{+}$ could overlap, potentially predicting incorrect probabilities for entangled measurements. Lewis et al. suggest their model can be extended to composite systems by relaxing PIP, allowing correlated ontic states, but they provide no explicit construction demonstrating that it reproduces quantum mechanics’ entanglement measurement predictions.

Our proof relies on the PBR theorem’s $\psi$-onticity for $\psi_1 \otimes \psi_2$, grounded in quantum mechanics’ entanglement measurement predictions, and the tensor product structure to establish subsystem $\psi$-onticity. The Lewis et al. model’s limitation lies in its failure to demonstrate that it can reproduce these predictions for composite systems. Without such a demonstration, it does not challenge the PBR theorem’s conclusion or our proof’s assertion that $\psi_1$ and $\psi_2$ are $\psi$-ontic for their subsystems. 

\subsection{Aaronson et al. Model}
Aaronson et al. propose a maximally nontrivial $\psi$-epistemic model for single quantum systems in any finite dimension $d$, extending the Lewis et al. approach by ensuring nontrivial overlap between epistemic distributions for all non-orthogonal states \cite{Aaronson2013}. This model aims to permit $\psi$-epistemicity by relaxing PIP, but its limitations, particularly its single-system focus and lack of a framework for composite-system entanglement measurements, mean it does not challenge our proof.

The Aaronson et al. model uses an ontic state space $\Lambda = \mathbb{CP}^{d-1} \times [0,1] \times \mathbb{N}$, where $\ket{\lambda} \in \mathbb{CP}^{d-1}$ is the $\psi$-related part, $x \in [0,1]$ is a continuous hidden variable, and an additional discrete variable in $\mathbb{N}$ ensures maximal nontriviality. The epistemic distributions $\mu_{\psi}(\lambda)$ are constructed as a convex combination of measures, allowing non-orthogonal states $\ket{\psi}$ and $\ket{\phi}$ to have overlapping supports with small total variation distance. Unlike symmetric models, where distributions are invariant under unitaries fixing $\ket{\psi}$, this model is explicitly nonsymmetric, leading to a complex, “unnatural” measure on the ontic space. The model reproduces the Born rule for single-system projective measurements, but it does not specify response functions for composite-system measurements, particularly entangled ones.

Like the Lewis et al. model, the Aaronson et al. model is designed for single systems. The PBR theorem relies on quantum mechanics’ predictions for entangled measurements in composite systems, such as Bell-basis measurements on states like $\ket{0} \otimes \ket{0}$ and $\ket{0} \otimes \ket{+}$. The Aaronson et al. model does not provide a framework for such measurements, leaving unclear how joint states are represented or how response functions handle entangled bases. Without demonstrating that it reproduces quantum mechanics’ predictions for composite systems by relaxing PIP, the model does not challenge our proof of $\psi$-ontology for single quantum systems.

\section{Criticisms and Responses}
\label{sec:criticisms}

\subsection{Hidden Variables and Correlations}

\textbf{Criticism:} The proof asserts that the decomposition of the delta distribution enforces \(\psi\)-onticity for subsystems, even in the presence of hidden variables \(\eta\). However, if hidden variables introduce correlations between subsystems, the disjointness of subsystem distributions might not hold. For example, if \(\nu_{\psi_1 \otimes \psi_2}(\eta)\) encodes correlations, the subsystem distributions \(\mu_{\psi_1}(\lambda_1)\) and \(\mu_{\psi_2}(\lambda_2)\) could overlap despite the delta distributions for \(\lambda_{\psi_1}\) and \(\lambda_{\psi_2}\).

\textbf{Response:} The PBR theorem establishes that \(\ket{\psi_1 \otimes \psi_2}\) is \(\psi\)-ontic, with epistemic distribution \(\mu_{\psi_1 \otimes \psi_2}(\lambda) = \delta(\lambda_{\psi} - \psi_1 \otimes \psi_2) \nu_{\psi_1 \otimes \psi_2}(\eta)\) (Section~\ref{sec:proof}). The delta distribution decomposes as \(\delta(\lambda_{\psi} - \psi_1 \otimes \psi_2) = \delta(\lambda_{\psi_1} - \psi_1) \delta(\lambda_{\psi_2} - \psi_2)\), fixing \(\lambda_{\psi_1} = \psi_1\) and \(\lambda_{\psi_2} = \psi_2\). Correlations in \(\eta\), encoded in \(\nu_{\psi_1 \otimes \psi_2}(\eta)\), cannot alter \(\lambda_{\psi_1}\) or \(\lambda_{\psi_2}\), as the delta functions enforce strict equality. Thus, for distinct states \(\ket{\psi_1} \neq \ket{\phi_1}\), $\mu_{\psi_1}(\lambda_{\psi_1}) = \delta(\lambda_{\psi_1} - \psi_1)$ and $\mu_{\psi_2}(\lambda_{\psi_2}) = \delta(\lambda_{\psi_2} - \psi_2)$ have disjoint supports, ensuring \(\psi\)-onticity, as stated in Section~\ref{sec:proof}. Hidden variables affect only \(\eta\), not the \(\psi\)-related parts, preserving disjointness regardless of correlations.

\subsection{Role of PIP}

\textbf{Criticism:} The proof claims to avoid PIP, but the decomposition $\mu_{\psi_1 \otimes \psi_2}(\lambda) = \delta(\lambda_{\psi_1} - \psi_1) \delta(\lambda_{\psi_2} - \psi_2) \nu_{\psi_1 \otimes \psi_2}(\eta)$ seems to implicitly assume a form of independence or separability in the ontic states of the subsystems.

\textbf{Response:} The decomposition is not an assumption of independence but a deduction from two established elements, as clarified in Section~\ref{sec:proof}. First, the PBR theorem establishes that the joint state is $\psi$-ontic, forcing the epistemic distribution into the form $\delta(\lambda_{\psi} - \psi_1 \otimes \psi_2) \nu_{\psi_1 \otimes \psi_2}(\eta)$. Second, for $\lambda_{\psi}$ to be uniquely identified with the product state $\ket{\psi_1 \otimes \psi_2}$, it must structurally encode the assignment of $\psi_1$ to subsystem 1 and $\psi_2$ to subsystem 2. This forces the decomposition $\lambda_{\psi} = (\lambda_{\psi_1}, \lambda_{\psi_2})$ with $\lambda_{\psi_1}=\psi_1$ and $\lambda_{\psi_2}=\psi_2$. 
Our derivation makes no auxiliary assumptions like PIP. The hidden variables $\eta$ in $\nu_{\psi_1 \otimes \psi_2}(\eta)$ can be arbitrarily correlated between subsystems, accounting for any possible non-separability in the full physical state, without affecting the $\psi$-onticity of the subsystems, which is determined solely by the fixed, disjoint $\psi$-related components. 

\subsection{Summary}

Joint $\psi$-onticity of product states, combined with the tensor-product structure of quantum mechanics, mathematically forces each subsystem to carry its own uniquely determined $\psi$-label. 
Hidden-variable correlations remain free but cannot modify the $\psi$-related components of the ontic state. 
Thus single-system $\psi$-ontology follows directly from the PBR theorem and the structure of composite quantum states, with no appeal—explicit or implicit—to PIP. 

\section{Implications for $\psi$-Ontology}
\label{sec:discuss}

Our result challenges a widespread assumption in the foundations of quantum mechanics: that a composite system can be jointly $\psi$-ontic while its individual subsystems remain $\psi$-epistemic, provided hidden-variable correlations between the subsystems are sufficiently strong~\cite{Leifer2014}.

We have shown that this possibility is incompatible with the tensor-product structure of quantum theory once the PBR theorem is taken at face value.  
As soon as product states $\ket{\psi_1 \otimes \psi_2}$ are $\psi$-ontic, the corresponding epistemic distributions necessarily take the form
\begin{equation}
\mu_{\psi_1\otimes\psi_2}(\lambda)
= \delta(\lambda_\psi - \psi_1\otimes\psi_2)\,
\nu_{\psi_1\otimes\psi_2}(\eta),
\end{equation}
where $\lambda=(\lambda_\psi,\eta)$ is a parametrization of the ontic space (see Appendix~\ref{app:representation}). 
Because the quantum state $\psi_1 \otimes \psi_2$ itself factorizes, the label $\lambda_\psi$ inherits the same tensor structure:
\begin{equation}
\lambda_\psi = (\lambda_{\psi_1},\lambda_{\psi_2}),
\end{equation}
and the delta constraint factorizes as
\begin{equation}
\delta(\lambda_\psi - \psi_1\otimes\psi_2)
= \delta(\lambda_{\psi_1} - \psi_1)\,
\delta(\lambda_{\psi_2} - \psi_2).
\end{equation}
Consequently, the marginal epistemic distributions for the individual subsystems are
\begin{equation}
\mu_{\psi_1}(\lambda_{\psi_1}) = \delta(\lambda_{\psi_1} - \psi_1),
\qquad
\mu_{\psi_2}(\lambda_{\psi_2}) = \delta(\lambda_{\psi_2} - \psi_2).
\end{equation}
These supports are manifestly disjoint whenever $\ket{\psi_1} \neq \ket{\phi_1}$ (resp.~$\ket{\psi_2} \neq \ket{\phi_2}$).  
Thus each subsystem is $\psi$-ontic, irrespective of arbitrary correlations that may exist in the residual variables $\eta$.

Subsystem $\psi$-epistemicity is therefore ruled out purely by the combination of (i) the PBR theorem’s conclusion that product states are jointly $\psi$-ontic and (ii) the tensor-product structure of the Hilbert space.  
No PIP and no restriction on hidden-variable correlations are required.

A common source of confusion in the literature has been the belief that non-trivial distributions over the residual variables $\eta$ could “smear out” the sharp encoding of the subsystem states.  
This conflates the full ontic state $\lambda = (\lambda_\psi,\eta)$ with the $\psi$-related component $\lambda_\psi$.  
Our analysis cleanly separates the two: $\lambda_\psi$ is uniquely fixed by the prepared quantum state, while $\eta$ may remain correlated and distributed in any way consistent with quantum predictions.  
Once this distinction is recognised, it becomes clear that the wave function is already sharply encoded in every ontic state, even in the presence of subsystem correlations.

Finally, a brief historical note on the decomposition $\lambda = (\lambda_\psi, \eta)$ is in order. 
Although the idea that a $\psi$-ontic model must contain a sharp quantum-state label is implicit in the original PBR definition of a physical property (see their Fig.~1 and accompanying text) \cite{PBR2012}, the explicit measure-theoretic decomposition of the ontic space into a $\psi$-label and residual degrees of freedom was not emphasised in early discussions.  
Subsequent work tended to focus on the operational role of PIP rather than on the structural consequences of disjoint supports.  
Moreover, identifying the wave function itself (or the ray in projective Hilbert space) as part of the ontic state was sometimes dismissed as a mere notational choice.  
Our proof elevates this identification from notation to necessity: once supports are pairwise disjoint, the decomposition $\lambda=(\lambda_\psi,\eta)$ and the resulting delta-function form are not additional assumptions, but the mathematically natural representation of $\psi$-onticity itself.

\section{Conclusions}

We have presented a minimal and completely general proof that joint $\psi$-onticity of product states, combined with the tensor-product structure of quantum theory, is already sufficient to establish $\psi$-onticity for individual quantum systems. While the PBR theorem is commonly understood, in its standard formulation, as deriving single-system $\psi$-onticity by invoking the Preparation Independence Postulate (PIP), our analysis shows that this further assumption is in fact unnecessary.

Once the ontic state of a composite system sharply encodes the product state $\ket{\psi_1 \otimes \psi_2}$ through a label $\lambda_\psi = \psi_1 \otimes \psi_2$ as required by joint $\psi$-onticity, the tensor factorization of the quantum state forces this label itself to decompose as $\lambda_\psi = \lambda_{\psi_1} \otimes \lambda_{\psi_2}$ with $\lambda_{\psi_1} = \psi_1$ and $\lambda_{\psi_2} = \psi_2$. The resulting marginal epistemic distributions for the subsystems are therefore delta-functions peaked on their respective wave functions, yielding disjoint supports for distinct subsystem states and hence full $\psi$-onticity.

Crucially, the argument never invokes PIP or places any restriction on possible correlations among additional hidden variables. Proposed $\psi$-epistemic models for single systems (e.g., Lewis et al.~\cite{Lewis2012}, Aaronson et al.~\cite{Aaronson2013}) avoid this conclusion only by failing to reproduce the entanglement-measurement statistics required for the PBR theorem to establish joint $\psi$-onticity in the first place. 
Our result thus resolves a longstanding conceptual confusion regarding the role of PIP in $\psi$-ontology theorems. It shows that once joint $\psi$-onticity for product states is granted, the quantum wave function must be regarded as a physical property of individual quantum systems, not merely of composite systems or statistical ensembles.


\appendix
\section{Mathematical Representation Theorem for $\psi$-Ontic Models}
\label{app:representation}\label{app:delta}

This appendix provides the rigorous mathematical foundation for the representation of $\psi$-ontic models used in the main text. 
In particular, we show that whenever epistemic distributions associated with distinct quantum states have disjoint supports, the ontic space admits a reparameterization in which the component identifying the quantum state is represented by a Dirac measure. Thus, the ``delta-function'' representation employed in the main text is not an additional assumption but a direct mathematical consequence of $\psi$-onticity. We also prove that this delta-function form naturally leads to a decomposition of the ontic state into a sharp quantum-state label and residual degrees of freedom. 

\begin{theorem}
\label{thm:delta-rigorous}
Let $(\Lambda, \Sigma)$ be a measurable space and let 
$\Psi = \mathbb{CP}^{n-1}$ be the projective Hilbert space equipped 
with its Borel $\sigma$-algebra $\mathcal{B}(\mathbb{CP}^{n-1})$.
Suppose $\{\mu_\psi\}_{\psi\in\Psi}$ is a family of probability measures on $(\Lambda,\Sigma)$ with pairwise disjoint supports:
\[
\supp(\mu_{\psi}) \cap \supp(\mu_{\phi}) = \varnothing
\qquad\text{for all }\psi\neq\phi.
\]
Then there exists a $\Sigma$-measurable function
\[
f\colon \Lambda \longrightarrow \Psi
\]
such that, for every $\psi\in\Psi$,
\[
\mu_{\psi}\bigl( f^{-1}(\{\psi\}) \bigr) = 1.
\]
Equivalently, one can write
\[
\mu_\psi(\lambda) = \delta\bigl( f(\lambda) - \psi \bigr) \, \nu_\psi(\eta),
\]
and define a decomposition of the ontic state
\[
\lambda = (\lambda_\psi, \eta), \qquad \lambda_\psi := f(\lambda) = \psi,
\]
where $\eta$ parametrizes points within $f^{-1}(\{\psi\})$.
\end{theorem}

\begin{proof}
Let $A_{\psi} := \supp(\mu_{\psi})\in\Sigma$.  
By the $\psi$-onticity assumption, the sets $\{A_\psi\}_{\psi\in\Psi}$ are pairwise disjoint.   
Define a measurable function $f:\Lambda\to\Psi$ by
\[
f(\lambda) :=
\begin{cases}
\psi  & \text{if }\lambda\in A_\psi \text{ for some }\psi,\\
\psi_0  & \text{if }\lambda\notin \bigcup_\psi A_\psi,
\end{cases}
\]
where $\psi_0 \in \Psi$ is an arbitrary fixed state.  
Since each $A_\psi$ is measurable, the preimage $f^{-1}(\{\psi\}) = A_\psi$ is measurable, so $f$ is $\Sigma$-measurable.  
Now, the complement $\Lambda \setminus A_\psi$ is a $\mu_\psi$-null set by the very definition of the support.  
Hence $\mu_\psi(A_\psi) = 1$ and
\[
\mu_\psi\!\bigl(f^{-1}(\{\psi\})\bigr) = \mu_\psi(A_\psi) = 1.
\]

Next, we define the residual degrees of freedom.  
For each $\psi$, the set $f^{-1}(\{\psi\}) = A_\psi$ contains all ontic states compatible with $\psi$.  
Let $\eta$ denote any measurable parameterization of points in $f^{-1}(\{\psi\})$, i.e., a measurable map $\lambda \mapsto \eta(\lambda)$.  
Define the conditional measure $\nu_\psi$ on $f^{-1}(\{\psi\})$ as the normalized restriction of $\mu_\psi$:
\[
\nu_\psi(B) := \mu_\psi\bigl( B \,\big|\, f^{-1}(\{\psi\}) \bigr),
\qquad B \subseteq f^{-1}(\{\psi\}) \text{ measurable.}
\]
Since $\mu_\psi(f^{-1}(\{\psi\})) = 1$, this conditional measure satisfies
\[
\nu_\psi\bigl(f^{-1}(\{\psi\})\bigr) = 1.
\]
Thus, for $\mu_\psi$-almost every $\lambda$, we may write
\[
\lambda = (\lambda_\psi, \eta), \qquad \lambda_\psi := f(\lambda) = \psi,
\]
where $\eta$ labels the residual degrees of freedom inside the fiber $f^{-1}(\{\psi\})$. 
With this normalization, the delta-function representation follows:
\[
\mu_\psi(\lambda_\psi, \eta) = \delta(\lambda_\psi - \psi)\, \nu_\psi(\eta),
\]
which is exactly the representation used in the main text.

Hence, any $\psi$-ontic model can be represented in the form of a delta function over the $\psi$-component together with a normalized conditional measure over the residual variables, making the structure of the ontic space explicit without adding extra assumptions.
\end{proof}

\begin{corollary}[Joint $\psi$-onticity for product states]
The PBR theorem establishes that, in any ontological model reproducing the predictions of quantum mechanics, the family of preparation measures
$\{\mu_{\psi_1 \otimes \psi_2}\}_{\psi_1 \in \Psi_1,\; \psi_2 \in \Psi_2}$ corresponding to all product states $|\psi_1\rangle \otimes |\psi_2\rangle$ has pairwise disjoint supports in the composite ontic space $\Lambda$. 
Applying Theorem~\ref{thm:delta-rigorous} to this family, there exists a measurable function
\[
f : \Lambda \longrightarrow \Psi_{12} \equiv \mathbb{CP}^{(d_1 d_2)-1},
\]
where $\Psi_{12}$ is the projective Hilbert space of the composite system, such that, for every product state $|\psi_1\rangle \otimes |\psi_2\rangle$,
\[
\mu_{\psi_1 \otimes \psi_2}\bigl(f^{-1}(\{\psi_1 \otimes \psi_2\})\bigr) = 1.
\]
Equivalently, up to $\mu_{\psi_1 \otimes \psi_2}$-null sets,
\[
\mu_{\psi_1 \otimes \psi_2}(d\lambda)
= \delta\bigl(f(\lambda) - \psi_1 \otimes \psi_2\bigr)\,
  \nu_{\psi_1 \otimes \psi_2}(d\eta),
\]
or, defining $\lambda_\psi := f(\lambda)$,
\[
\lambda = (\lambda_\psi, \eta), 
\qquad
\lambda_\psi = \psi_1 \otimes \psi_2 \quad (\text{$\mu_{\psi_1 \otimes \psi_2}$-almost everywhere}).
\]
Thus, every ontic state $\lambda$ of the composite system carries a sharp, uniquely determined label $\lambda_\psi$ that exactly identifies the prepared product state. This provides a precise structural foundation for the subsequent derivation of subsystem $\psi$-onticity: the decomposition of $\lambda_\psi$ into subsystem components follows directly from the tensor-product form of $\psi_1 \otimes \psi_2$.
\end{corollary}

The physical meaning of the delta-function representation is as follows. In any $\psi$-ontic model, the probability measures associated with distinct quantum states have disjoint supports in the ontic space $\Lambda$. This disjointness is the operational content of $\psi$-onticity: no ontic state can be compatible with more than one quantum state. Once the supports are disjoint, the ontic space can be reorganized by defining a measurable map $f:\Lambda \to \Psi$ that assigns to each ontic state the unique quantum state $\psi$ whose support contains it. By definition, $f(\lambda)$ is uniquely determined by $\lambda$ and varies nontrivially across different preparations; it can thus be interpreted as an ontic variable. One can then introduce coordinates $(\lambda_\psi, \eta)$ on $\Lambda$ by setting $\lambda_\psi := f(\lambda)$ to label the quantum state and letting $\eta$ parametrize the residual degrees of freedom within the fiber $f^{-1}(\psi)$. In these coordinates, the preparation measure takes the delta-function form $\mu_\psi(\lambda_\psi, \eta) = \delta(\lambda_\psi - \psi)\, \nu_\psi(\eta)$, showing that the ontic state contains a sharp component $\lambda_\psi = \psi$ that uniquely identifies the prepared quantum state, while $\eta$ accounts for additional ontic degrees of freedom. This construction is a direct consequence of the disjointness of supports and the definition of the ontological-model framework; it requires no assumptions beyond $\psi$-onticity.

\end{document}